\documentclass[aps,twocolumn,superscriptaddress,floatfix,preprintnumbers]{revtex4}%{{{
\usepackage[english]{babel}
\usepackage[utf8]{inputenc}
\usepackage{amsmath}
\usepackage{amssymb}
\usepackage{amsthm}
\usepackage{mathtools}
\usepackage{dsfont}
\usepackage{graphicx}
\usepackage{soul}

\usepackage[draft,inline,nomargin]{fixme} \fxsetup{theme=color}

\usepackage{xcolor}

\newcommand{\tr}{{\rm tr}\,}
\newcommand{\ket}[1]{\left|{#1}\right\rangle}
\newcommand{\bra}[1]{\left\langle{#1}\right|}
\newcommand{\braket}[2]{\langle{#1}|{#2}\rangle}

\newcommand{\ketbra}[2]{\left|{#1}\rangle\!\langle{#2}\right|}

\newtheorem{theorem}{Theorem}
\newtheorem{lemma}{Lemma}

\newtheorem{definition}{Definition}

\begin{document}
\title{Stability Verification of Quantum non-i.i.d. sources}
\date{\today}
\author{Esteban Mart\'inez Vargas}
\email{estebanmv@protonmail.com}
\affiliation{Calle centella 3, Smza 18 Mza 3, 77505 Cancún México}
%\author{Gael Sent\'is}
%\email{Gael.Sentis@uab.cat}
%\affiliation{F\'isica Te\`orica: Informaci\'o i Fen\`omens Qu\`antics, Departament de F\'isica, Universitat Aut\`onoma de Barcelona, 08193 Bellatera (Barcelona) Spain}
%\author{Michalis Skotiniotis}
%\email{m.skotiniotis@gmail.com}
%\affiliation{F\'isica Te\`orica: Informaci\'o i Fen\`omens Qu\`antics, Departament de F\'isica, Universitat Aut\`onoma de Barcelona, 08193 Bellatera (Barcelona) Spain}
%\author{Ramon Mu\~noz-Tapia}
%\email{Ramon.Munoz@uab.cat}
%\affiliation{F\'isica Te\`orica: Informaci\'o i Fen\`omens Qu\`antics, Departament de F\'isica, Universitat Aut\`onoma de Barcelona, 08193 Bellatera (Barcelona) Spain}
%\author{John Calsamiglia}
%\email{John.Calsamiglia@uab.cat}
%\affiliation{F\'isica Te\`orica: Informaci\'o i Fen\`omens Qu\`antics, Departament de F\'isica, Universitat Aut\`onoma de Barcelona, 08193 Bellatera (Barcelona) Spain}
%
%
%}}}
\begin{abstract}%{{{
    We introduce the problem of stability verification of quantum sources 
    which are non-i.i.d.. The problem consists in ascertaining whether a given quantum source is
    stable or not, in the sense that it produces always a desired quantum state
    or if it suffers deviations. Stability is a statistical notion related to the sparsity of errors.
    This problem is closely related to the problem of quantum verification first
    proposed by Pallister et. al. \cite{OptimalVerificPallis2018}, however,
    it extends the notion of the original problem. We
    introduce a family of states that come from these non-i.i.d. sources which we call
    a Markov state. These sources are more versatile than the i.i.d. ones as they allow statistical
    deviations from the norm instead of the more coarse previous approach. 
    We prove in theorem \ref{thm:stvsMkv} that the Markov states are not
    well described with tensor products over a changing source. In theorem \ref{thm:lowbound}
    we further provide
    a lower bound on the trace distance between two Markov states, 
    or conversely, an upper bound on the fidelity between these states. This is a bound
    on the capacity of determining the stability property of the source, which shows that
    it is exponentially easier to ascertain this with respect to $n$, the number of outcomes
    from the source.
\end{abstract}
\maketitle
%}}}
\section{Introduction}%{{{
Quantum tomography is the process of reconstructing a quantum state
from a series of observations \cite{Nielsen2011QCQ}. 
This is a very costly
process \cite{ScalableMultipHaffne2005} as it normally requires an exponential amount of
measurements with the dimension of the system  \cite{PredictingManyHuang2020}, 
which implies an exponential amount of copies of the state. Alternative approaches
have been invented to circumvent this issue: using compressed sensing for example
\cite{QuantumStateTGross2010}. Recently, there have been interesting lines of 
research whose objective is less ambitious than full-state tomography, but to 
calculate functionals of states that take a polynomial amount of resources
\cite{PredictingManyHuang2020,ShadowTomograpAarons2018}. 

Close to this topic is the task of quantum verification 
\cite{OptimalVerificPallis2018}, whose objective is to ascertain if a source yields
a desired state, or if it incurs an error. The question to answer is if a
machine that produces identical copies of the state $\ket{\psi}$ and 
whose details are hidden from us (is a black box) is producing the 
state it should. Here one does not deal with the full
tomographical problem and therefore number of required measurements can be 
lower.

Pallister et. al. \cite{OptimalVerificPallis2018} define verification as a quantum
hypotheses testing problem, which consists in guessing a given quantum state
from two possible hypotheses with the lowest probability of error \cite{QuDiscandApps2017Bae,BarnettQSD09}.
The task is simple to state: suppose a machine produces states $\{\sigma_1,\sigma_2,\ldots,\sigma_n\}$ 
which should be $n$ copies of $\ket{\psi}$. Hypothesis 0 is that 
$\sigma_i=\ketbra{\psi}{\psi}$ for all $i$ and hypothesis 1 is that 
$\bra{\psi}\sigma_i\ket{\psi}\leq 1-\epsilon$ for all $i$ for $0\leq\epsilon\leq1$.
The objective is that the verifier passes the test with a worst-case probability
of $\delta$. They consider independent online measurements 
\cite{Sentis2022online}. 

Despite making considerable advances, their approach is an oversimplification as it
restricts to detect very specific situations: the source produced all the time the correct state or all the time a wrong one.
Perhaps one would qualify as not so bad a machine that 
produces a desired state $\ket{\psi}$ \emph{most} of the time, but here and there, in
a sparse manner, allow an error.

Part of the oversimplification of the problem lies in the fact that its 
definition uses i.i.d. sources that limit the 
abstract notions that one would like to address.
Here, we extend the notion of verification of Pallister from identifying
outputs of an i.i.d. source. Instead of aiming for detecting a perfectly 
consistent source, 
we allow deviations as long as they are statistically negligible. 
We introduce the formalism of a family of mixed states that describe 
rigorously this situation. These states are prepared by
a source that is non-i.i.d. but has temporal correlations between the produced
states. We shall call the sources we study ``Markov sources'' as their 
definition depends on Markov chains. 

We investigate how the family of states that we introduce behaves and we find
that the tensor product of states after each iteration does not apply in this 
case. In some sense, the Markov sources we introduce here
generalize the notion of a Markov chain to quantum systems, although studied
through other approaches \cite{sutter2018approximate}. 
We show in section \ref{sec:windows} the
relationship between the sparsity of errors for the Markov source and
two parameters $\epsilon\in[0,1]$ and $\delta\in[0,1]$.
Then, we arrive at theorem \ref{thm:stvsMkv},
which shows the difference between the Markov source and a similar one
is exponential in the number of instances of the Markov source. 

Having defined the Markov sources and their respective output after $n$ 
instances, we address the problem of verification which can be translated into
a quantum discrimination problem. Pallister's approach uses individual measurements
for several copies, therefore their measuring scheme has no horizon. 
By contrast, our approach presupposes that a fixed number of instances of the Markov source is
given. We observe that this problem can be translated into two hypotheses:
$H_0$ is that we were given a state $\mathfrak{M}_0=f(\delta_0)$ 
and $H_1$ that $\mathfrak{M}_1=g(\delta_1)$
where $0\leq\delta_0<\delta_1\leq1$. We have two quantum states, therefore the problem reduces
to find the Helstrom bound between $\mathfrak{M}_0$ and $\mathfrak{M}_1$.
We give a lower bound for the trace distance between these states in 
theorem \ref{thm:lowbound}, or conversely, an upper bound on the fidelity between these states. 
We end this paper with a discussion.
%}}}
\section{Quantum Markov sources}%{{{
Here we introduce the formalism for describing the non-i.i.d. sources
that we study here.
This kind of sources can be defined in analogy with iid sources of
mixed states.
Suppose that we have a pair of orthogonal states $\ket{0}$ and $\ket{1}$.
Also, suppose that a source yields the state $\ket{0}$ with probability $p$,
and the state $\ket{1}$ with probability $1-p$. An iid source would 
produce after $n$ instances a state given by $\rho^{\otimes n}$ with
\begin{equation}
    \rho=p\ketbra{0}{0}+(1-p)\ketbra{1}{1}.
\end{equation}
Let us denote $\textbf{x}_n$ a string of size $n$ composed of $k$ $0$'s and
$n-k$ $1$'s. Observe that
\begin{equation}
    \bra{\textbf{x}_n}\rho^{\otimes n}\ket{\textbf{x}_n}=p^{k}(1-p)^{n-k}.
\end{equation}
For a different string $\ket{\textbf{y}_n}$ we have $\bra{\textbf{y}_n}\rho^{\otimes n}\ket{\textbf{x}_n}=0$.
Then, we can think of $\rho^{\otimes n}$ as a convex combination of pure 
states projectors:
\begin{equation}
    \rho^{\otimes n}=\sum_{\textbf{x}_i\in \mathcal{S}_i}\mathcal{T}(\textbf{x}_i)\ketbra{\textbf{x}_i}{\textbf{x}_i}
    \label{eq:tensorrho}
\end{equation}
where $\mathcal{S}_i$ denotes the set of permutations of $0$'s and $1$'s in
strings of size $i$ and $\mathcal{T}(\textbf{x}_i)=p^{k}(1-p)^{i-k}$.

We note that the state $\rho^{\otimes n}$ depends on the function 
$\mathcal{T}(\textbf{x}_i)$ that assigns a probability to the string $\textbf{x}_i$.

Consider now a source of quantum states that produces  
a concatenation of the states $\ket{0}$ and $\ket{1}$ following
the outcome of a 2-state Markov chain. This means that we have a machine
with an interior state $\{0\}$ or $\{1\}$ which produces a state accordingly.
This machine is analogous to the one illustrated in Fig. (\ref{fig:ssdiag}) 
with $0\leq\varepsilon\leq1$ and $0\leq\delta\leq1$.
\begin{figure}
    \center
    \includegraphics[width=0.49\textwidth]{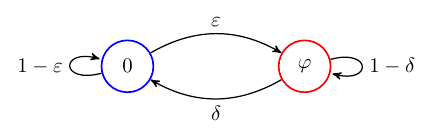}
    \caption{Diagram of a non-i.i.d. source governed by a two-state Markov chain.}
    \label{fig:ssdiag}
\end{figure}

At the start, we are given the state $|0\rangle$
with probability $p_0$ and the state $|1\rangle$ with probability $1-p_0$.
Therefore, we are given the mixed state
\begin{equation}
    \mu_1=p_0\ketbra{0}{0}+(1-p_0)\ketbra{1}{1}.
\end{equation}
Now, as our source follows the Markov chain of the figure (\ref{fig:ssdiag}) the
state after another output is
\begin{align}
    \mu_2&=p_0(\varepsilon\ketbra{01}{01}+(1-\varepsilon)\ketbra{00}{00})\nonumber\\
    &+(1-p_0)(\delta\ketbra{10}{10}+(1-\delta)\ketbra{11}{11}).
    \label{eq:2MarkovS}
\end{align}
From equation (\ref{eq:2MarkovS}) it is straightforward to see that 
$\mu_2=\mu_1\otimes\mu_1$ if and only if $\varepsilon+\delta=1$ and
$\delta=p_0$ which is a special case of the parameters at hand.

It is simple to construct the density matrix that is the output of $n$ steps
of the source from equation (\ref{eq:2MarkovS}).
We denote the state of the output of $N$ instances of a Markov source as in figure (\ref{fig:ssdiag})
as $\mathfrak{M}[\varepsilon,\delta,\ket{0},\ket{1},n]$. Therefore,
analogously to equation (\ref{eq:tensorrho}) we have 
\begin{equation}
    \mathfrak{M}[\varepsilon,\delta,\ket{0},\ket{1},n]:=\sum_{\textbf{x}_i\in \mathcal{S}_i}\mathcal{M}(\textbf{x}_i)\ketbra{\textbf{x}_i}{\textbf{x}_i}
   \label{eq:bigM}
\end{equation}
for a suitable function $\mathcal{M}(\textbf{x}_i)$, that follows the 
probabilities of the Markov chain for a given string $\textbf{x}_i$.
Being more general, we have the following definition
\begin{definition}
    \label{def:MarkovSt}
    We denote an $n$-length Markov state that is the mix of the  
    pure states $\ket{0}$ and $\ket{\varphi}$ with $\braket{0}{\varphi}=c$, $0\leq c\leq1$ as
\begin{equation}
    \mathfrak{M}[\varepsilon,\delta,\ket{0},\ket{\varphi},n]:=\sum_{\chi_i\in \mathcal{S}_i}\mathcal{M}(\chi_i)\ketbra{\chi_i}{\chi_i},
    \label{eq:bigMphi}
\end{equation}
where the states $\ket{\chi^{\prime}_i}$ are (tensored) strings of $0$'s and
$\varphi$'s that correspond to all the permutations of $n$ of n instances of 2 symbols.
The function $\mathcal{M}(\chi_i)$ correspond to the probabilities from a
Markov chain as in Fig. (\ref{fig:ssdiag}).
\end{definition}
Observe that
$\mathfrak{M}[\varepsilon,\delta,\ket{0},\ket{\varphi},n]\in\mathcal{H}^{\otimes n}$
however clearly in general $\mathfrak{M}[\varepsilon,\delta,\ket{0},\ket{\varphi},n]\neq\mu_1^{\otimes n}$.

\subsection{Sparsity: Bound on $k-$size windows with $l$ errors}
\label{sec:windows}
    A question arises when using the states defined by Eq. (\ref{eq:bigMphi}). 
    We want to assign meaning to the values $\epsilon$ and $\delta$.
    We want to verify that the non-i.i.d. source does not accumulate too many 
    errors.

    Suppose that for a string of $N$ instances of the Markov source
    we want for any possible set of $k$ consecutive states to
    have $l$ errors or less. This can be done probabilistically,
    with a probability $\alpha$. 

    For example, we want that any consecutive set of $k=3$ to have
    at most $l=1$ error. We have to consider cases of consecutive sets
    which overlap. Considering the full string of size $N$ we have
    that the string with the maximum possible number of errors that 
    start with a 0 is given by

    \begin{equation}
        0~\varphi~0~0~\varphi~0~0\varphi~0~0~\varphi\ldots
        \label{eq:maxerrors}
    \end{equation}
    
    Any other string which fulfills the stable requirement would have more
    $0$'s between the $\varphi$'s. Analyzing the probability of the string
    (\ref{eq:maxerrors}) we have that it is given by
    \begin{equation}
        p_0\epsilon\delta(1-\epsilon)\epsilon\delta(1-\epsilon)\epsilon\delta(1-\epsilon)\ldots=p_0(\epsilon\delta)^{\lfloor (N-1)/3\rfloor}(1-\epsilon)^{\lfloor (N-1)/3\rfloor}.
    \end{equation} 
    Where $p_0$ is the probability that we are given $0$ at the beginning. Therefore,
    the probability that we get all the stable strings is given by the sum
    \begin{align}
        S(\delta,\epsilon)\equiv\sum_{j=0}^{\lfloor (N-1)/3\rfloor}&\binom{\frac{N-1}{3}}{\frac{N-1}{3}-j}p_0(\epsilon\delta)^{\lfloor (N-1)/3-j\rfloor}\nonumber\\
        &\times(1-\epsilon)^{\lfloor (N-1)/3+2j\rfloor},
        \label{eq:epsilondeltas}
    \end{align}
    because if we substract $j$ errors then we add two transitions from $0$ to $0$.
    Observe that we consider all the possible combinations that come from taking $(N-1)/3-j$ 
    errors away.
    We can graph the function (\ref{eq:epsilondeltas}) and maximize it. Numerical
    maximization reveals something that we could have guessed: the maximum is
    attained when $\epsilon\rightarrow0$ and $\delta\rightarrow1$. In Fig. (\ref{fig:numeric})
    we have a graph for $S(\delta,\epsilon)$ defined in Eq. (\ref{eq:epsilondeltas}).
    \begin{figure}
        \includegraphics[width=0.49\textwidth]{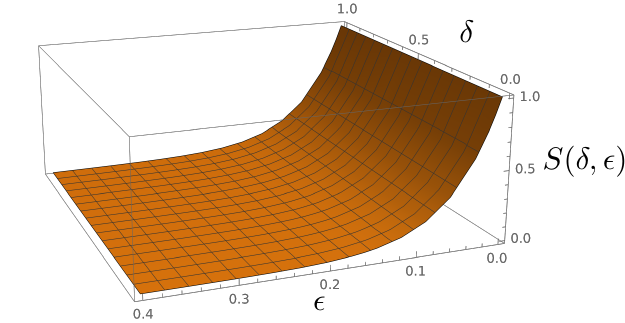}
        \caption{Graph of $S(\delta,\epsilon)$ for $N=20$ and we got $0$ at the start, i.e. $p_0=1$.}
        \label{fig:numeric}
    \end{figure}

    The generalization of this procedure to any $k$-size window with $l$-errors is straightforward.
    We have to find the limiting cases for the number of errors first. Then we can obtain the probabilities
    of all other cases that fulfill this stability criterion and the sum is the searched probability.
\subsection{Consistency}%{{{
The Markov sources we just described ask for a different relationship between each
iteration of the machine. We note that the i.i.d. source produces a \emph{tensored}
state. However, the tensor product will not describe accurately the Markov source.
Let us investigate how the Markov source would work.

Suppose we have the state $\ket{0}$
with probability $p$ and the state $\ket{\varphi}$ with probability $1-p$.
This would be equivalent to begin with a vector distribution 
\begin{equation}
    \mu=
\begin{pmatrix}
p \\
1-p
\end{pmatrix}.
\end{equation}
Then, depending on being on the state $\{0\}$ or $\{\varphi\}$ in the machine,
it will jump with probabilities $\varepsilon$ and $\delta$ to be in the 
aforementioned states. The distribution for the states of the machine is  
given by $\mu\mathbf{P}$ where
\begin{equation}
    \mathbf{P}=
\begin{pmatrix}
    1-\varepsilon& \varepsilon \\
    \delta & 1-\delta
\end{pmatrix}.
\end{equation}
After 2 iterations we would have the distribution $\mu\mathbf{P}^2$ etc. Let us
call $\rho_k$ the state that has the distribution $\mu\mathbf{P}^k$. 
A $n$ output of the source is then $\rho_1\otimes\rho_2\otimes\ldots\otimes\rho_n$.
We will call these last kind of sources \emph{tensored} Markov sources.
Observe then that $\rho_k\rightarrow\rho_{st}$ when $k\rightarrow\infty$ where
\begin{equation}
    \rho_{st} = \frac{\delta}{\varepsilon+\delta}\ketbra{0}{0}+\frac{\varepsilon}{\varepsilon+\delta}\ketbra{\varphi}{\varphi}.
    \label{eq:stationary}
\end{equation}
We will call ``stationary'' the density matrix $\rho_{st}$ in terms of $\varepsilon$ and
$\delta$. In the beginning, we would have some density matrix that would evolve
into the stationary density matrix in a finite number of steps and then stay
in that state, therefore the difference becomes constant while the number of 
states can be arbitrarily large.
We observe in the following lemma that the output of $n$ instances of the
tensored Markov source should be similar to $\rho_{st}^{\otimes n}$ when $n\rightarrow\infty$.
\begin{lemma}
    The fidelity between the output of $n$ instances of the tensored Markov source
    and the correspondent stationary density matrix is constant for $n\rightarrow\infty$.
\end{lemma}
\begin{proof}
    Because of the multiplicativity over tensor products of the fidelity (property 9.2.5
    in \cite{wilde_2013}) then,
    \begin{equation}
        \mathcal{F}(\rho_1\otimes\rho_2\otimes\ldots\otimes\rho_n,\rho_{st}^{\otimes n})=\prod_{i=1}^n\mathcal{F}(\rho_i,\rho_{st}).
        \label{eq:fids}
    \end{equation}
    Clearly, for any $\tau>0$ there exist an $n_0$ large enough so that the right-hand side of 
    eq. (\ref{eq:fids}) 
    \begin{equation}
        |\prod_{i=1}^n\mathcal{F}(\rho_i,\rho_{st})-\prod_{i=1}^{n+1}\mathcal{F}(\rho_i,\rho_{st})|<\tau
\end{equation}
    for any $n>n_0$.
\end{proof}
We show in theorem 1 that the fidelity between the Markov source and $\rho_{st}^n$ 
decays exponentially and the Markov source is 
therefore not well described by the formalism of the tensor product.
First, we state a lemma for the Fidelity \cite{wilde_2013}.
\begin{lemma}
    The fidelity between the states $\rho$ and $\sigma$ defined as
\begin{equation}
    \mathcal{F}(\rho,\sigma)=\left(\tr\left(\sqrt{\sqrt{\sigma}\rho\sqrt{\sigma}}\right)\right)^2.
\end{equation}
is equal to the minimum Bhattacharya distance between two probability distributions $p_m$ and
$q_m$ resulting from a measurement $\{\Gamma_m\}$ of the states $\rho$ and $\sigma$:
\begin{equation}
    \mathcal{F}(\rho,\sigma)=\min_{\{\Gamma_m\}}\left(\sum_m \sqrt{p_mq_m}\right)^2,
\end{equation}
with
\begin{equation}
    p_m=\tr[\Gamma_m\rho],\quad q_m=\tr[\Gamma_m\sigma].
\end{equation}
\label{thm:bhatta}
\end{lemma}

We can then prove the following theorem,
\begin{theorem}
    \label{thm:stvsMkv}
    For any $\tau$, $0\leq\tau\leq1$ there exist $n$ sufficiently large such that,
    \begin{align}
        \log_2&(\mathcal{F}(\mathfrak{M}[\varepsilon,\delta,\ket{0},\ket{\varphi},n],\rho_{st}^{\otimes n}))\leq\nonumber\\
        &-n(H(\delta/(\delta+\varepsilon))+H(\mathfrak{M})-2\tau\nonumber\\
        &-2*|H(\delta/(\delta+\varepsilon))-H(\mathfrak{M})+2\tau|).
    \end{align}
    where $H(\delta/(\delta+\varepsilon))$ denotes the entropy for the binomial
distribution with parameter $p=\delta/(\delta+\varepsilon)$ and we have the
entropy rate 
    \begin{equation}
        H(\mathfrak{M}):=\frac{\delta}{\varepsilon+\delta}H(\varepsilon)+\frac{\varepsilon}{\varepsilon+\delta}H(\delta).
    \end{equation}
\end{theorem}
\begin{proof} 
    The distribution of the state (\ref{eq:stationary})
    corresponds to a binomial distribution and asymptotically, the set of strings that
    correspond to the typical set $A_\tau^{(n)}$ will be given by
    \begin{equation}
        2^{-n(H(\delta/(\delta+\varepsilon))+\tau)}\leq p(\textbf{x}_n)\leq2^{-n(H(\delta/(\delta+\varepsilon))-\tau)}.
        \label{eq:typicalA}
    \end{equation}
    The Markov state $\mathfrak{M}[\varepsilon,\delta,\ket{0},\ket{\varphi},n]$ instead will     
    be discribed by an entropy rate $H(\mathfrak{M})$.
    Therefore, the typical set $B^{(n)}_\tau$ of sequences for these strings will have a probability between
    \begin{equation}
        2^{-n(H(\mathfrak{M})+\tau)}\leq q(\textbf{x}_n)\leq2^{-n(H(\mathfrak{{M}})-\tau)}.
        \label{eq:typicalB}
    \end{equation}
    Therefore, following theorem 3.1.2 from \cite{cover2006} 
    the intersection of the typical sets has a cardinality 
    \begin{equation}
        |A^{(n)}_\tau\cap B^{(n)}_\tau|=2^{n|H(\delta/(\delta+\varepsilon))-H(\mathfrak{M})+2\tau|}. 
    \end{equation}
    Using lemma (\ref{thm:bhatta}) we can calculate the Fidelity by
    optimizing over POVMs $\{\Gamma_m\}$. If we do not optimize over POVMs
    and take a specific measurement, then we have an upper bound for the
    fidelity. The specific POVM we take is a von Neumann measurement on the
    Hilbert space $\mathcal{H}^{\otimes n}$: the projectors $\ketbra{\textbf{x}_n}{\textbf{x}_n}$.
    Note that the set $\ket{\textbf{x}_n}$ forms an orthonormal basis of
    $\mathcal{H}^{\otimes n}$, nevertheless there is no reason for this to be the 
    optimal POVM, so in general we got an upper bound for the Fidelity.
    Using the asymptotic behavior of the typical sets we can see how
    the Fidelity behaves in the large $n$ limit \cite{cover2006}.
    Let us define the function $g(\chi_n)$ as the number of $\varphi$'s in 
    the string $\chi_n$. Using the right-hand side of equations (\ref{eq:typicalA}) and (\ref{eq:typicalB}) 
    we have that
    \begin{align}
        \bra{\textbf{x}_n}&\mathfrak{M}[\varepsilon,\delta,\ket{0},\ket{\varphi},n]\ket{\textbf{x}_n}\bra{\textbf{x}_n}\rho^{\otimes n}_{st}\ket{\textbf{x}_n}\leq\nonumber\\
        &c^{4*g(\chi_n)}2^{-n(H(\delta/(\delta+\varepsilon))+H(\mathfrak{M})-2\tau)}
        \label{eq:pqineq}
    \end{align}
    Using the Lemma (\ref{thm:bhatta}) we thus have the following upper bound for the 
    Fidelity,
    \begin{align}
        \mathcal{F}&\leq\left(\sum_{\chi_n\in A^{(n)}_\tau\cap B^{(n)}_\tau}c^{4*g(\chi_n)}2^{-\frac{n}{2}(H(\delta/(\delta+\varepsilon))+H(\mathfrak{M})-2\tau)}\right)^2\\
        &\leq \left(\sum_{\chi_n\in A^{(n)}_\tau\cap B^{(n)}_\tau}2^{-\frac{n}{2}(H(\delta/(\delta+\varepsilon))+H(\mathfrak{M})-2\tau)}\right)^2\\
        &= 2^{-n(H(\delta/(\delta+\varepsilon))+H(\mathfrak{M})-2\tau-2*|H(\delta/(\delta+\varepsilon))-H(\mathfrak{M})+2\tau|)}. 
    \end{align}
    Where we have used that the typical set is uniformly distributed.
\end{proof}

In summary, in the limit $n\rightarrow\infty$ the Fidelity between the
Markov state (\ref{eq:bigMphi}) and the stationary state
(\ref{eq:stationary}) decays exponentially.
The reason is that the typical set
of a binary distribution with $p=\delta/(\epsilon+\delta)$ is different
then the typical set of the outcome of a Markov chain after $n$ instances.
In Fig. (\ref{fig:fidelio}) we have an numerical graph for the 
nonorthogonal state $\ket{\varphi}$ such that $\arccos(c)=2\pi/3$.
%}}}
\begin{figure}
    \includegraphics[width=0.49\textwidth]{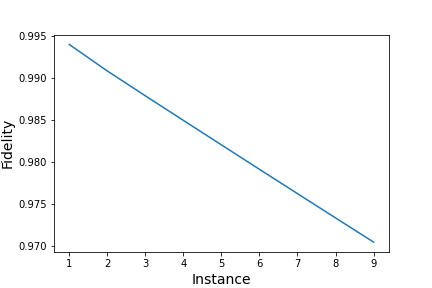}
    \caption{Fidelity in the vertical axis, the number of instances
        in the horizontal axis. Here $p_0=0.5$, $\varepsilon=0.3$, $\delta=0.5$ and $\theta=\pi/3$.}
    \label{fig:fidelio}
\end{figure}
%}}}
\section{Stability verification of non-i.i.d. sources}%{{{
The problem of stability verification of a non-i.i.d. source can be easily stated now.
Given $\mathfrak{M}[\varepsilon,\delta,\ket{0},\ket{\varphi},n]$ infer the
value of $\varepsilon$ and $\delta$. 
Perhaps the full estimation of this density matrix is somewhat costly
and we can restrict ourselves to a simpler task.
We fix the value of $\varepsilon=\varepsilon_0$ and have two hypotheses, 
$\{\delta\leq\delta_0,\delta\geq\delta_1\}$.
The optimal worst-case scenario probability of success is given by the
Helstrom bound \cite{wiseman_milburn_2009} between the states $\mathfrak{M}_0:=\mathfrak{M}[\varepsilon_0,\delta_0,\ket{0},\ket{\varphi},n]$ 
and $\mathfrak{M}_1:=\mathfrak{M}[\varepsilon_0,\delta_1,\ket{0},\ket{\varphi},n]$.

To calculate the Helstrom bound we need to calculate the trace distance between the states
$\mathfrak{M}_0$ 
and $\mathfrak{M}_1$. 
We further show that the trace distance can be bounded analytically.
The following relationship between the trace distance and the fidelity will be useful \cite{wilde_2013}
\begin{equation}
    1-\sqrt{\mathcal{F}(\rho,\sigma)}\leq\frac{1}{2}\left\Vert\rho-\sigma\right\Vert_1.
    \label{eq:fidtrace}
\end{equation}

\begin{theorem}
    \label{thm:fidbound}
    The trace distance between the states $\mathfrak{M}_0$ and $\mathfrak{M}_1$ is lower bounded by
    \begin{equation}
    \frac{1}{2}\left\Vert \mathfrak{M}_0-\mathfrak{M}_1\right\Vert\geq 1-\sqrt{\mathcal{B}(\mathfrak{M}_0,\mathfrak{M}_1)}. 
\end{equation}
where        
\begin{equation}
\mathcal{B}(\mathfrak{M}_0,\mathfrak{M}_1):=2^{-n(H(\mathfrak{M}_0)+H(\mathfrak{M}_1)-2\tau-2*|H(\mathfrak{M}_0)-H(\mathfrak{M}_1)+2\tau|)}.
\end{equation} 
\label{thm:lowbound}
\end{theorem}
\begin{proof}
    From Lemma (\ref{thm:bhatta}) we can obtain an upper bound for fidelity using the same arguments
    as for theorem (\ref{thm:stvsMkv}), which means picking a specific POVM. 
    \begin{align}
        \bra{\textbf{x}_n}&\mathfrak{M}_0\ket{\textbf{x}_n}\bra{\textbf{x}_n}\mathfrak{M}_1\ket{\textbf{x}_n}\leq\nonumber\\
        &c^{4*g(\chi_n)}2^{-n(H(\mathfrak{M_0})+H(\mathfrak{M_1})-2\tau)}
        \label{eq:pqMMsineq}
    \end{align}
    Following analogous steps of the proof of theorem \ref{thm:stvsMkv} we have
    \begin{align}
        \mathcal{F}&\leq\left(\sum_{\chi_n\in A^{(n)}_\tau\cap B^{(n)}_\tau}c^{4*g(\chi_n)}2^{-\frac{n}{2}(H(\mathfrak{M}_0)+H(\mathfrak{M}_1)-2\tau)}\right)^2\\
        &\leq \left(\sum_{\chi_n\in A^{(n)}_\tau\cap B^{(n)}_\tau}2^{-\frac{n}{2}(H(\mathfrak{M}_0)+H(\mathfrak{M}_1)-2\tau)}\right)^2\\
        &= 2^{-n(H(\mathfrak{M}_0)+H(\mathfrak{M}_1)-2\tau-2*|H(\mathfrak{M}_0)-H(\mathfrak{M}_1)+2\tau|)}.
    \end{align}
    Using relationship (\ref{eq:fidtrace}) we finally arrive at the lower bound for the trace distance
    \begin{equation}
        \frac{1}{2}\left\Vert \mathfrak{M}_0-\mathfrak{M}_1\right\Vert\geq 1-\sqrt{\mathcal{B}(\mathfrak{M}_0,\mathfrak{M}_1)}. 
    \end{equation}
\end{proof}
As in the case of theorem \ref{thm:stvsMkv} here the fidelity decays exponentially but with a different
exponent. Observe that our bound tells us that when the fidelity is very small then the trace distance goes to 1. 

We can calculate numerically the trace distance 
between the operators $\mathfrak{M}_0$ and $\mathfrak{M}_1$.
Fixing the values of $p_0=0.5,~\varepsilon=0.3,$ $\delta_0=1-\delta$ and
$\delta_1=\delta$ we have the behavior for the trace distance in 
Fig. (\ref{fig:MStrace}) for $N=7$.

\begin{figure}
    \center
    \includegraphics[width=0.49\textwidth]{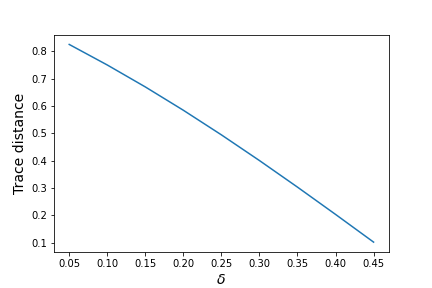}
    \caption{We show the trace distance between two Markov states
        with $p_0=0.5,~\varepsilon=0.3,$ $\delta_0=1-\delta$, 
        $\delta_1=\delta$ and the number of instances is $N=7$.}
    \label{fig:MStrace}
\end{figure}
%}}}
\section{Discussion}%{{{
We introduce the stability verification of quantum non-i.i.d. sources problem 
which is closely related to the verification of states
introduced by Pallister et. al \cite{OptimalVerificPallis2018}. The 
stability verification problem requires
the introduction of a type of non-i.i.d. sources that we study here.
We state and show a theorem that shows that the Markov sources 
are not well described by the tensor product of a source that 
gradually goes to a stationary output state. We then turn to the
problem of stability verification of quantum non-i.i.d. sources which consists of a
quantum state discrimination problem. We obtain a lower bound on the
trace distance between the quantum hypotheses in this problem. Conversely, this
is an upper bound on the fidelity of such states, this is obtained in theorem (\ref{thm:fidbound}).

The family of states that we study here is versatile as the
problem of change point \cite{QuantumChangeSentis2016} can be written in terms of these states.
The kind of sources we introduce would be those that self-regulate
somehow, we would be detecting a servo system that corrects a 
faulty output.
A full study (parametrization) of the Markov states introduced here is
the subject of future study.

Notice that theorem (\ref{thm:stvsMkv}) implies that there is a mathematical
difference between the two ways of looking at a source. First, we might think the source
momentaneously produces a density matrix that is changing at each step it
outputs a state. The other kind of states considers all the strings possible
by following the simple rule of a Markov chain and then weights the strings
according to a Markov chain. Theorem (\ref{thm:stvsMkv}) tells us that these
two kind of states differ. What is happening is that there are temporal 
correlations that are ignored when the formalism of tensor products is taken
but that becomes evident in the Markov state formalism.

An interesting perspective would be to think about quantum algorithms that involve
Markov states. One might speculate that there can be algorithms that use states
from definition (\ref{def:MarkovSt}) that show
quantum advantages over classical computers. These would constitute a new family of
algorithms with quantum advantage.
%}}}
\section{Acknowledgements}%{{{
I thank Gael Sent\'is for giving the idea from which this research originates 
and for useful discussions on this topic.
I thank Ram\'on Mu\~noz-Tapia for useful discussions on this topic.%}}}
\bibliography{MarkovSourcebib}%{{{
\end{document}